\documentclass{article}

\usepackage[a4paper,margin=3cm]{geometry}

\usepackage[ngerman,english]{babel}
\usepackage[utf8]{inputenc}
\usepackage{setspace}
\usepackage{xcolor}
\usepackage{amsmath,amsthm,amssymb,amsfonts}
\usepackage{graphicx}
\usepackage{psfrag}
\usepackage{physics}
\usepackage{changebar}
\usepackage{comment}
\usepackage{listings}
\usepackage{enumitem}
\usepackage{mathtools}
\usepackage{mathrsfs}
\usepackage[scr=boondox,cal=esstix]{mathalpha}
\usepackage{algorithm2e}
\usepackage{multirow}
\usepackage[sort&compress,comma,square,numbers]{natbib}
\usepackage[bf]{caption}
\RestyleAlgo{ruled}

\newtheorem{Def}{Definition}[section]

\newtheorem{Lemma}[Def]{Lemma}
\newtheorem{Prop}[Def]{Proposition}

\newtheorem{Kor}[Def]{Corollary}

\newtheorem{Theo}[Def]{Theorem}

\usepackage{hyperref}

\usepackage{LesiMacros}

\title{\sffamily\bfseries Characterisation of reactive Nash equilibria in repeated additive games}

\date{}

\author{
\parbox[c]{14cm}{
\centering
\onehalfspacing
\fontsize{12}{12}\selectfont
Franziska Lesigang$^{1}$, Christian Hilbe$^{1}$, Nikoleta E. Glynatsi$^{2, 3}$\\[0.4cm]
$^1$Interdisciplinary Transformation University, Linz, Austria.\\
$^2$RIKEN Center for Computational Science, Kobe, Japan.\\
$^3$RIKEN Center for Interdisciplinary Theoretical and Mathematical Science (iTHEMS), Wako, Japan.\\[0.15cm]
}
}

\begin{document}

\maketitle

\section*{Abstract}
In this paper, we study reactive strategies in repeated additive games between
two players with finitely many actions. Reactive strategies condition only on
the opponent's previous action, making them one of the simplest ways players
can respond to past interactions. Additive games include important models of
cooperation, such as the donation game and games with a punishment option.
We show that, for this class of games and strategies, the conditions for
symmetric Nash equilibria reduce to a system of linear equalities and
inequalities in the strategy parameters, allowing us to characterise all such
equilibria. We establish a one-to-one correspondence between non-empty subsets
$S$ of the action set and equilibrium classes, which we call
$S$-supporting equilibria. These are equilibria that use exactly the actions in
$S$ when playing against themselves. As a special case, we recover the well-known equalizer
strategies as the equilibria supported on the entire action set.
To assess which equilibrium classes are most evolutionarily
relevant, we complement our analytical characterisation with simulations of
social learning dynamics. We find that their prevalence is determined by two
factors: how likely they are to be generated and how robust they are against
invasion.\\

\noindent
\textbf{keywords:} iterated game, symmetric Nash equilibria, reactive strategies, additive games

\section{Introduction}

Repeated games are a widely used framework in game theory to study how strategic
interactions unfold over
time~\cite{Axelrod:Science:1981,mailath:book:2006,Garcia:FRAI:2018}. In such
games, a player's behaviour in the present may depend on what happened in the
past. This allows repeated games to capture central features of social
interaction, including reciprocity, cooperation and
punishment~\cite{sigmund:book:2010,Ohtsuki:Punishment:2009}. In many settings,
players are assumed to remember only a finite number of previous moves. That is,
players employ \emph{memory-$n$} strategies, with $n$ being the number of past
rounds
remembered~\cite{hauert:PRSB:1997,martinez-vaquero:plosone:2012,stewart:scirep:2016}.
A central question is: which of these strategies form a symmetric Nash
equilibrium? That is, what are the strategies such that, when adopted by
everyone, no player has an incentive to deviate?

In general, verifying whether a given repeated game strategy is a Nash
equilibrium is non-trivial. After all, for such a verification one needs to
study all (uncountably many) deviations. In the special case of memory-$n$
strategies, however, we can build on several previous studies that simplify this
task~\cite{Press:PNAS:2012,Levinsky:EconResearch:2010,Glynatsi:PNAS:2024,
Laporte:PNAS:2026}. For example, the results of
Levinsky~\emph{et al.}\cite{Levinsky:EconResearch:2010} imply that, against a
focal memory-$n$ strategy, there always exists a best response among the pure
(i.e., deterministic) memory-$n$ strategies. In particular, if none of these
pure memory-$n$ strategies yields a larger deviation payoff, then the focal
memory-$n$ strategy is an equilibrium. Because the set of pure memory-$n$
strategies is finite, so is the number of deviations that need to be tested.

An analogous result also applies to \emph{reactive-$n$} strategies (i.e., those
memory-$n$ strategies that only depend on the opponent's past $n$ actions). For
those strategies, one can always find a best response among the pure
\emph{self-reactive-$n$} strategies -- those memory-$n$ strategies that only
depend on the player's own past $n$
actions~\cite{Glynatsi:PNAS:2024}. A further simplification is possible for
\emph{additive games}. A stage game is additive if the players' payoffs can be
decomposed into two summands, each depending on only one player's
action~\cite{McAvoy:PLOSCB:2021}. A canonical example is the donation game.
Here, cooperation imposes a cost on the acting player while providing a benefit
to the opponent. More generally, repeated additive games have the property that
total payoffs only depend on how often each player chooses an action, rather
than on how often each pair of actions occurs. For additive games, one can show
that any reactive-$n$ strategy has a best response among the pure
self-reactive-$(n\!-\!1)$
strategies~\cite{Lesigang:EconLett:2025}. This observation further reduces the
complexity of verifying whether a reactive-$n$ strategy forms a symmetric Nash
equilibrium.

Yet, even when deviations are tractable, characterising Nash equilibria remains
challenging because it requires an explicit calculation of a strategy's
self-payoff. This calculation becomes increasingly difficult in games with many
actions. Consequently, most previous characterisations of Nash equilibria have
either focused on games with two actions only, such as the repeated prisoner's
dilemma~\citep{stewart:pnas:2014,park:NComms:2022}, or restricted attention to
special subclasses of strategies for which self-payoffs are easy to compute,
such as the \emph{partner strategies}~\cite{Hilbe:partners:2015}.
Even though much work has focused on the prisoner's dilemma, there is growing
interest in social dilemmas with three or more
actions~\cite{ZHANG:Multiaction:2026,ZHANG:ThreeAction:2026,
Ohtsuki:Punishment:2009}. One example is the work of
Rand~\emph{et al}.~\cite{Ohtsuki:Punishment:2009}, who study a donation game
with an additional punishment option. Among other results, they characterise
the game's Nash equilibria among the reactive strategies.

In this paper we characterise \emph{all} symmetric Nash equilibria among the
reactive-$n$ strategies with $n\!=\!1$, which we simply refer to as
\emph{reactive}
strategies~\citep{nowak:AMC:1989,nowak:Nature:1992a,imhof:PRSB:2010,
allen:AmNat:2013}. Our characterisation applies to arbitrary additive games
with any finite number of actions and is based solely on linear conditions on
the strategy parameters.
Our characterisation reveals a simple structure of the equilibrium set. We
establish a one-to-one correspondence between non-empty subsets of the action
set and different subclasses of Nash equilibria, which we call $S$-supporting
equilibria. These equilibria use exactly the actions in $S$ when interacting
with themselves. Furthermore, if one player deviates towards another strategy,
and if that other strategy also only uses actions in $S$, the deviating
player's payoff is left invariant. That is, against a strategy that forms an
$S$-supporting equilibrium, \emph{any} strategy that only uses actions in
$S$ yields the same payoff. In the special case when $S$ is the full action
set, we recover the well-known class of equalizer
strategies~\cite{boerlijst:AMM:1997,Press:PNAS:2012,Hilbe:equalizers:2013}.

Having established this classification, we turn to the question of which
equilibrium classes are most robust from an evolutionary perspective.
Specifically, we use an evolutionary process of
imitation~\cite{traulsen:PRE:2006b} to study the stability of the different
equilibrium classes, and their ability to emerge under social learning. We
focus on a specific example: a donation game with an additional action that is
moderately cooperative. This action imposes a smaller cost than full
cooperation and provides the opponent with a smaller benefit. The results
indicate that the most robust equilibrium classes are those supported on small
sets $S$. In particular, such classes are more likely to emerge and tend to be
more stable once established.

\section{Model}

\subsection{Game setup}

We consider a two-player repeated game without discounting. 
In each round, a player (referred to as the \emph{focal player}) chooses an action out of a
finite action set $A$. 
Similarly, the other player (the \emph{opponent}) chooses an
action from a finite set $B$. 
Although the action sets
coincide in our main result, distinguishing them in some proofs helps clarify
the roles of the different players. We assume that the resulting payoffs are
defined by the function $g\!: A\!\times\! B \rightarrow \R^2$. We assume payoffs to
be additive~\cite{McAvoy:PLOSCB:2021}. 
That is, there exist functions $g_A\!: A \rightarrow \R^2$ and $g_B\! : B
\rightarrow \R^2$ such that $g = g_A + g_B$.
Further, we denote $g = (g^i, g^{-i})^T$, where $g^i$ is the focal player's
payoff and $g^{-i}$ is the opponent's payoff.

Players choose each round what to do by implementing a strategy -- a rule that determines which action to take for any previous history of play.
One important subclass of strategies is memory-$1$ strategies~\cite{sigmund:book:2010}. That is, the players' moves in each round only depend on the outcome of the previous round. Formally, a memory-$1$ strategy $p$ is a vector 
\begin{align}
(p_{ab})_{(a,b)\in A \times B} \\
p_{ab} \in \Delta^{|M|-1},
\end{align}
where $M \in \set{A,B}$ is the player's action set and $\Delta^{m-1} \coloneqq \set{x \in \R^m_{\geq 0}: \sum_{i\leq m} x_i = 1}$.
We say the focal player's strategy $p$ is \emph{reactive} if for all $a_1,a_2 \in A$ and $b \in B$
\begin{align}
    p_{a_1b} = p_{a_2b}.
\end{align}
That is, the focal player's behaviour depends only on the opponent's previous action. Accordingly, we write $p_b(a)$ for the focal player's probability to play action $a$ after observing the opponent play $b$.
Similarly, the strategy $p$ is \emph{self-reactive} if for all $a \in A$ and $b_1,b_2 \in B$
\begin{align}
    p_{ab_1} = p_{ab_2}.
\end{align}
That is, players with a self-reactive strategy only consider their own past action.
The strategy~$p$ is \emph{unconditional} if it is independent of the player's previous actions.
That is, for all $a_1,a_2 \in A$ and for all $b_1,b_2 \in B$
\begin{equation}
    p_{a_1b_1} = p_{a_2b_2}.
\end{equation}
Furthermore, a strategy is \emph{pure} if for any previous game outcome, it prescribes playing a single action with certainty (rather than randomizing over two or more actions). 
We denote by $Alla$ the pure unconditional strategy that always plays action $a \in A$.
Analogous definitions apply to the opponent with action set $B$.

\subsection{Calculation of payoffs}

Let $\pi^i_t(p, \tilde{p})$ denote the focal player's expected payoff in round $t$, given that the focal player and the opponent use strategies $p$ and $\tilde{p}$, respectively. The player's \emph{long-term payoff} is defined as
\begin{equation}\label{eq::long-term-payoff}
    \pi^i(p, \tilde{p}) \coloneqq \lim_{\tau \rightarrow \infty} \frac{1}{\tau} \sum_{t = 1}^{\tau} \pi^i_t(p,\tilde{p}).
\end{equation}
Analogously, we can define the opponent's long term payoff $\pi^{-i}$, and we set $\pi \coloneqq (\pi^i, \pi^{-i})^T$.
While the limit in~\eqref{eq::long-term-payoff} may not exist for general strategies, it does exist if both players use memory-$1$ strategies. In that case, however, it may depend on the moves in the first round, when no history has yet been observed. In the following, we only specify these first-round moves when necessary. 

To calculate long-term payoffs, we represent the game as a Markov chain. 
The states of the Markov chain are all possible outcomes of a given round, $A \times B$. 
Because players make their decisions independently, the probability of moving from one state $a_1b_1$ to another state $a_2b_2$ is given by
\begin{align}
    M_{a_1b_1,a_2b_2} = p_{a_1b_1}(a_2) \cdot \tilde{p}_{a_1b_1}(b_2).
\end{align}
For $t \geq 1$, let $v_t \in \Delta^{|A|\times |B|-1}$ denote the probability distribution of observing any possible game outcome in round $t$. For a given distribution $v_t$, we compute the next round's distribution as $v_{t+1} = v_t M$. Since $M$ is a stochastic matrix, the Theorem of Perron-Frobenius ensures that $v_t$ converges to a limiting distribution $v \in \Delta^{|A|\times |B|-1}$. This limiting distribution is a left eigenvector of $M$ with respect to the eigenvalue~$1$. If $M$ is also primitive, the limiting distribution is unique. Otherwise it is again uniquely determined by the players' initial actions.

We note that the focal player's expected payoff can be written as
\begin{align}
    \pi^i_t(p, \tilde{p}) = v_t\cdot u.
\end{align}
Here $u = (g^i(a,b))_{(a,b)\in A\times B}$ is the vector of all possible stage-game payoffs. Furthermore, if 
\begin{align}
v_\tau \overset{\tau \rightarrow \infty}{\rightarrow} v
\end{align} 
then so does $\frac{1}{\tau} \sum_{t \leq \tau} v_t$. We obtain

\begin{equation}
    \pi^i(p,\tilde{p}) = \lim_{\tau \rightarrow \infty} \frac{1}{\tau} \sum_{t = 1}^{\tau} \pi^i_t(p,\tilde{p}) 
    =\lim_{\tau \rightarrow \infty} \frac{1}{\tau} \sum_{t = 1}^{\tau} v_t \cdot u 
    = \left(\lim_{\tau \rightarrow \infty} \frac{1}{\tau} \sum_{t = 1}^{\tau} v_t \right) \cdot u = v \cdot u
\end{equation}
The opponent's long-term payoff $\pi^{-i}(p, \tilde{p})$ can be calculated similarly. 
In the following, we will use the notation $v(a,b)$, which describes the abundance of rounds in which the focal player chooses action~$a$ whereas the opponent chooses action $b$.

\subsection{Symmetric games and Symmetric Nash equilibria}\label{subsec::Model::Nash-equilibira}

In the following, we restrict attention to symmetric games, such that $A =B$. 
A \emph{best response} to a strategy $p$ is a strategy $q$ that maximizes a player's long-term payoff against $p$. 
That is, $q$ is a best response if  $\pi^i(q,p) \geq \pi^i(\sigma,p)$ for every strategy $\sigma$.
A \emph{symmetric Nash equilibrium} is a strategy $p$ satisfying
\begin{equation}
    \pi^i(p,p) \geq \pi^i(\sigma,p)
\end{equation}
for all strategies $\sigma$. That is, $p$ is a best response to itself. Throughout the paper, whenever we refer to a Nash equilibrium, we mean a symmetric Nash equilibrium in the above sense.

Our earlier results in Ref.~\cite{Lesigang:EconLett:2025} imply that against any reactive
strategy, there exists a best response among the pure unconditional strategies.
Thus, to verify whether a reactive strategy~$p$ constitutes a Nash
equilibrium, it suffices to check if for all $a \in A$,
\begin{equation} \label{Eq:NashCondition}
    \pi^i(p,p)\geq \pi^i(Alla,p).
\end{equation}
The payoffs of pure unconditional strategies against a reactive strategy
depend linearly on $p$. 
The right-hand side of \eqref{Eq:NashCondition} is therefore straightforward to calculate. 
The main technical difficulty instead lies in the self-payoff term $\pi^i(p,p)$ on the left-hand side.
This payoff typically needs to be obtained by computing the stationary distribution corresponding to the respective $|A|^2\times |A|^2$ transition matrix. 
Deriving a suitable representation of the self-payoff of a reactive strategy
is therefore the key step toward characterizing all Nash equilibria among the reactive strategies.

\section{Results}

\subsection{Characterisation of reactive Nash equilibria}\label{sec:nash_equilibria}

Our main result is the characterisation of reactive Nash equilibria in repeated
additive games. Specifically, we prove the following theorem.

\begin{Theo}\label{Theo::Better_Characterisation_of_Nash}
    Let $p$ be a reactive strategy for a repeated additive game with action set $A$.
    Then $p$ is a Nash equilibrium if and only if there exists a non-empty subset $S \subseteq A$ such that
    \begin{enumerate}
        \item $p_b(a) = 0$ for all $a \in A\setminus S$ and $b \in S$,
        \item If the stationary distribution of the strategy $p$ against itself is not unique, the initial action must be chosen such that the resulting stationary distribution is supported on $S$,
        \item $\pi^i(Alla_1,p) = \pi^i(Alla_2,p)$ for all $a_1,a_2 \in S$,
        \item $\pi^i(Alla_1,p) \geq \pi^i(Alla_2,p)$ for all $a_1 \in S$ and all $a_2 \in A\setminus S$.
    \end{enumerate}
\end{Theo}

\noindent
The first two conditions specify that $S$ consists of those actions that $p$
plays against itself. In the following, we show that these conditions are
equivalent to requiring that the stationary distribution of $p$ against itself
is supported on $S$. We therefore refer to the induced class of Nash equilibria
as $S$-supporting equilibria. According to the other two conditions, any pure
unconditional strategy that always plays the same action in $S$ must achieve the
best response payoff against~$p$; actions outside of $S$ yield a lower (or at
most equal) payoff. In fact, we show below that any opponent that restricts
their actions to $S$ receives the best response payoff, regardless of what
they condition on (even if they take more than the last round into account).
This further motivates the name $S$-supporting equilibria.

We emphasize that the above conditions are not only necessary, but also
sufficient. There are no other Nash equilibria apart from these described
classes. Suppose a Nash equilibrium is both $S_1$-supporting and
$S_2$-supporting for different $S_1$ and $S_2$. Let $S\coloneqq S_1 \cap S_2$.
Then the equilibrium is also $S$-supporting. To avoid assigning the same
equilibrium to multiple classes of Nash equilibria, we define each Nash
equilibrium to be $S$-supporting only for the minimal $S$. We thus obtain
disjoint classes of Nash equilibria that correspond to distinct subsets of the
action~set.

Interestingly, the above characterisation of Nash equilibria does not require
the calculation of self-payoffs $\pi^i(p,p)$. Moreover, since
$\pi^i(Alla,p) =
\sum_{\tilde{a}\in A}g^{-i}_A(\tilde{a}) p_a(\tilde{a}) + g^i_A(a)$
for any action $a$, all conditions in the characterisation are linear
equalities and inequalities in the strategy space. Thus, we provide a direct
and tractable method for computing all reactive Nash equilibria.\\

%-----------------------------------------------------------------------------------------
\noindent\textit{Proof.}
The proof of Theorem~\ref{Theo::Better_Characterisation_of_Nash} is split into
two main steps. We discuss the main ideas behind each step below, while the
proofs of all required lemmas can be found in the \textbf{Appendix}.

As discussed in Section~\ref{subsec::Model::Nash-equilibira}, a reactive strategy $p$ is a Nash equilibrium if and only if 
\begin{equation}\label{eq::reduced_nash_equation}
    \pi^i(p,p) \geq \pi^i(Alla,p)
\end{equation}
for all actions $a \in A$. Since the right-hand side of
Equation~\eqref{eq::reduced_nash_equation} is linear in $p$, the main difficulty
lies in analyzing the left-hand side. To avoid the explicit computation of
self-payoffs, we derive a representation that expresses the self-payoff as a
convex combination of the payoffs obtained against pure unconditional
strategies. The coefficients of this convex combination depend only on the
stationary distribution induced when $p$ plays against itself.
To obtain this representation of self-payoffs, we require two lemmas.

\begin{enumerate}
\item First, we show that any opponent strategy $\sigma$ facing $p$ can
be replaced by an unconditional strategy $q$ without changing either player's
long-run payoff. At this step, additivity of the game and $p$ being
reactive become essential. In additive games, payoffs depend only on how
frequently each player chooses each action, rather than on the frequencies of
particular action profiles. We therefore construct $q$ so that it plays each
action with the same frequency as $\sigma$. Since $p$ conditions only on the
opponent's current action frequencies and not on additional history,
replacing $\sigma$ by $q$ leaves the action frequencies of $p$ unchanged.
Therefore, the long-run payoffs remain unchanged. This result is formalized
in the following lemma, where we distinguish the players' action sets by $A$ and
$B$ for notational clarity.\\

\begin{Lemma}\label{lem::finding_unconditional} Let $p : B \rightarrow
    \Delta^{\left|A\right|-1}$ be an arbitrary but fixed reactive strategy
    and $\sigma : A^n \rightarrow \Delta^{\left|B\right|-1}$ be an arbitrary memory-$n$ opponent for some $n\in \N$ opponent. We assume that the initial first $n$ actions  of $\sigma$ are fixed but arbitrary. Let $v\in \Delta^{\left|A\right|\times\left|B\right|-1}$ where $v(a,b)$ is the frequency of the action pair $(a,b)$ being played by $p$ and $\sigma$ during the course of the game. Note that if $\sigma$ is a memory-$1$ strategy then $v$ is the stationary distribution of the game. We define an
    unconditional strategy $q$ via
    \begin{equation}
        q(b) \coloneqq \sum_{a \in A} v(a,b).
    \end{equation}
    Then
    \begin{equation}
        \begin{aligned}
            \begin{pmatrix}
                \pi^i(p,\sigma)\\
                \pi^{-i}(p,\sigma)
            \end{pmatrix} = 
            \begin{pmatrix}
                \pi^i(p,q)\\
                \pi^{-i}(p,q)
            \end{pmatrix}.
        \end{aligned}
    \end{equation}
\end{Lemma}

\item Next, we show that the payoff of an unconditional strategy against a
reactive opponent can be represented as a convex combination of the payoffs
of pure unconditional strategies.\\

\begin{Lemma}\label{lem::unconditional_payoff}
    Let $p : B \rightarrow \Delta^{\left|A\right|-1}$ be a reactive strategy and $q \in \Delta^{\left|B\right|-1}$ be an unconditional strategy. It holds that
    \begin{equation}\label{eq::convex_comb}
        \begin{aligned}
            \begin{pmatrix}
                \pi^i(p,q)\\
                \pi^{-i}(p,q)
            \end{pmatrix} = 
            \sum_{b\in B} q(b)
            \begin{pmatrix}
                \pi^i(p,Allb)\\
                \pi^{-i}(p,Allb)
            \end{pmatrix}.
        \end{aligned}
    \end{equation}
\end{Lemma}

\end{enumerate}

\noindent
Combining Lemma~\ref{lem::finding_unconditional} and
Lemma~\ref{lem::unconditional_payoff}, and setting $A = B$, immediately yields
Corollary~\ref{Kor::self-payoff} where we obtain our representation of
self-payoffs.

\begin{Kor}\label{Kor::self-payoff}
    Let $p$ be a reactive strategy. Let $v \in \Delta^{2|A|-1}$ be the stationary distribution of $p$ playing against itself. Then the self-payoff can be written as 
    \begin{align}
        \pi^i(p,p) = \sum_{a \in A} \left(\sum_{\tilde{a} \in A} v(a, \tilde{a})\right) \pi^i(Alla,p).
    \end{align}
\end{Kor}

As shown in~\cite{Lesigang:EconLett:2025}, there exists at least one action
$a \in A$ such that the corresponding pure unconditional strategy $Alla$
is a best response to $p$. For $p$ to be a Nash equilibrium, its self-payoff
must coincide with the payoff of a best response. Corollary~\ref{Kor::self-payoff}
establishes that the self-payoff can be expressed as a convex combination of the
payoffs of pure unconditional strategies. Therefore, every coefficient
corresponding to a payoff that is not maximal must be zero. Since these
coefficients are given by entries of the stationary distribution of $p$ against
itself, this determines the support of the stationary distribution. Formally, we
obtain the following lemma.

\begin{Lemma}\label{Lem::Characterisation_of_Nash}
    Let $p$ be a reactive strategy and $v$ be a stationary distribution of $p$ playing against itself. 
    Then $p$ is a Nash equilibrium if and only if there exists a non-empty subset $S \subseteq A$ such that
    \begin{enumerate}
        \item $\supp{v} \subseteq S\times S$, i.e., $v(a_1,a_2) = 0$ if $a_1 \in A\setminus S$ or $a_2 \in A\setminus S$,
        \item If the stationary distribution is not unique, the initial action must be chosen such that the resulting stationary distribution is in $S$.
        \item $\pi^i(Alla_1,p) = \pi^i(Alla_2,p)$ for all $a_1,a_2 \in S$,
        \item $\pi^i(Alla_1,p) \geq \pi^i(Alla_2,p)$ for all $a_1 \in S$ and $a_2 \in A\setminus S$.
    \end{enumerate}
\end{Lemma}

\noindent
Stationary distributions are in general costly to compute. Therefore, we require the following lemma, which replaces Condition~1 of Lemma~\ref{Lem::Characterisation_of_Nash} with a more tractable condition.

\begin{Lemma}\label{lem::Adopt_condition_3}
    If the $S$ in Lemma~\ref{Lem::Characterisation_of_Nash} is minimal, then we can replace the lemma's first condition by
    \begin{enumerate}
        \item[1.] $p_b(a) = 0$ for all $a \in A\setminus S$ and $b \in S$, 
    \end{enumerate}
\end{Lemma}

Combining these results concludes the proof of Theorem~\ref{Theo::Better_Characterisation_of_Nash}.

\subsection{Best responses to $S$-supporting equilibria}

As noted above, the first two conditions of
Theorem~\ref{Theo::Better_Characterisation_of_Nash} imply that an $S$-supporting
equilibrium $p$ uses only actions in $S$ when playing against itself. We next show that $p$
satisfies a second property: any opponent that restricts its play to
actions in $S$ when facing $p$ is a best response to $p$. In other words, $p$
rewards every such opponent with the best response payoff, regardless of the
frequencies with which the actions in $S$ are played. We formalize this property
in the following proposition.

\begin{Prop}\label{Prop::opponent-payoff}
    Let $p$ be an $S$-supporting Nash
    equilibrium and $\sigma$ be an arbitrary memory-$n$ opponent for some $n \in \N$. Let $v \in
    \Delta^{2|A|-1}$ denote the action frequencies of $\sigma$ playing
    against $p$. We assume that $\sigma$ only plays actions in $S$ against $p$
    in the long-run, i.e., $v(a_1,a_2) = 0$ if $a_1 \notin S$. Then, $\sigma$ is
    a best response to $p$. 
\end{Prop}

\noindent
In the special case that $S$ is the entire action set, $S\!=\!A$, the above result implies that any strategy yields the same payoff against an $A$-supporting equilibrium. 
Thus, $A$-supporting equilibria coincide with the class of equalizer
strategies~\cite{boerlijst:AMM:1997,Press:PNAS:2012,Hilbe:equalizers:2013}.

\begin{proof}
    It follows immediately from Lemma~\ref{lem::finding_unconditional} and Lemma~\ref{lem::unconditional_payoff} that
    \begin{equation}\label{eq::general_opponent_general_form}
        \pi^i(\sigma,p) = \sum_{a\in A} \left(\sum_{\tilde{a} \in A} v(a,\tilde{a})\right) \pi^i(Alla,p).
    \end{equation}
    Recall that $v(a,\tilde{a})$ denotes the long-run frequency of the state in
    which $\sigma$ plays action $a$ while $p$ plays action $\tilde{a}$. Since
    $\sigma$ only uses actions in $S$ in the long run,
    Equation~\eqref{eq::general_opponent_general_form} simplifies to
    \begin{equation}
    \pi^i(\sigma,p) = \sum_{s\in S}
    \left(\sum_{\tilde{a} \in A} v(s,\tilde{a})\right) \pi^i(Alls,p).
    \end{equation}
    Thus, the payoff of $\sigma$ is a convex combination of the payoffs induced
    by the pure unconditional strategies $Alls$ for $s \in S$. Since $p$ is an
    $S$-supporting equilibrium, every strategy $Alls$ with $s \in S$ attains
    the maximal payoff against $p$. Therefore, every term in the above convex
    combination equals the maximal payoff, implying that, for every $s \in S$,
    \begin{equation}
    \pi^i(\sigma,p) = \pi^i(Alls,p).
    \end{equation}
    In particular, $\sigma$ is a best response to $p$.
\end{proof}

\subsection{Dimensions of equilibrium classes}
The characterisation in Theorem~\ref{Theo::Better_Characterisation_of_Nash}
partitions the set of reactive Nash equilibria into classes indexed by subsets
of the action set. A natural question is how large these classes are. While the
existence of a given class depends on the underlying stage game, we can compare
classes in a game-independent way by studying their maximal dimensions. More
precisely, we ask how many free parameters remain once the conditions defining
an $S$-supporting equilibrium have been imposed.

We first note that the existence of an $S$-supporting equilibrium depends on the
payoff structure of the stage game. To see this, consider an arbitrary game with
action set $A$ and let $a_j,a_k\in A$. Assume that
$\max_{\sigma \in \Sigma}\pi(Alla_j,\sigma) < \min_{\sigma \in \Sigma}\pi(Alla_k,\sigma),$
where $\Sigma$ is the set of all strategies for which stationary
distributions against $Alla_j$ and $Alla_k$ exist. Then there exists no
strategy $p \in \Sigma$ such that $Alla_j$ receives a payoff against $p$ that
is at least as large as the payoff of $Alla_k$ against $p$. Consequently, no
$S$-supporting equilibrium can exist for any subset $S\subseteq A$ with
$a_k\in S$.

Despite this dependence on the stage game, the dimensions of equilibrium
classes can be bounded purely as a function of $|S|$ and the number of actions.
Let $m\coloneqq |A|$. We define the quantity $d_{|S|,m}$ to be the maximum
number of independent parameters that can be chosen within an
$S$-supporting equilibrium class. Larger values therefore correspond to larger
families of equilibria. The following proposition provides this upper bound.

\begin{Prop}\label{Prop::Degrees_of_freedom}
    Let $S \subseteq A$ be arbitrary but fixed and $m\coloneqq |A|$. The class of $S$-supporting Nash equilibria has at most
    \begin{equation}
        d_{|S|,m} = (|S|-1)^2 + (m-|S|)(m-1)
    \end{equation}
    degrees of freedom.
\end{Prop}

\begin{proof}
    For each $b \in S$, the vector $p_b$ assigns probability only to actions in
    $S$. Since $\sum_{a\in A} p_b(a) = 1$, the vector $p_b$ contains $|S|-1$
    free variables. Across all $b \in S$, this yields $|S|(|S|-1)$ free
    variables. Further, Condition~3 reduces the number of free variables by at
    least $|S|-1$, leaving
    \begin{equation}
        |S|(|S|-1)- (|S|-1) = (|S|-1)^2.
    \end{equation}
    For each $b \notin S$ we obtain another $m-1$ degrees of freedom, at most.
    Therefore, the total number of independent variables is
    \begin{equation}
        (|S|-1)^2 + (m-|S|)(m-1).
    \end{equation}
    We note that this is an upper bound since some inequalities in Condition~4
    might never hold for specific games. Then we receive less degrees of
    freedom.
\end{proof}

\noindent
Interestingly, the bound is symmetric in $|S|$ around $(m+1)/2$, implying that
classes supported on very small subsets and classes supported on very large
subsets have comparable dimensions. In particular, $d_{1,m}=d_{m,m}=(m-1)^2$.
Thus, the classes corresponding to partner-like equilibria and equalizers are
among the largest equilibrium classes.
Figure~\ref{fig::degrees_of_freedom_ten_actions} illustrates that equilibrium
classes of intermediate support are considerably more constrained than those
near the extremes. In this sense, equalizers ($|S|=|A|$) and single-action
equilibria ($|S|=1$) occupy the largest regions of the equilibrium space.

%%%%%%%%%%%%%%%%%%%%%%%
%% THREE ACTION DONATION GAME %%
%%%%%%%%%%%%%%%%%%%%%%%

\subsection{Three action donation game}
To illustrate Theorem~\ref{Theo::Better_Characterisation_of_Nash}, we consider a
three-action donation game in which players can choose between cooperation~($C$), moderation~($M$),
and defection~($D$). Cooperation provides a benefit~$b_1$ to the opponent at a
cost~$c_1$. Moderation provides a smaller benefit~$b_2$ at a smaller cost~$c_2$. 
Defection neither comes with a cost nor with a benefit. 
The resulting payoff matrix is
\begin{equation}
    \begin{gathered}
    \quad C \quad  \quad M  \quad \quad D
     \\
    \begin{matrix}
        C \\ M \\D
    \end{matrix}
    \begin{pmatrix}
        b_1-c_1 & b_2-c_1 &-c_1\\
        b_1-c_2 & b_2-c_2 & -c_2 \\
        b_1 & b_2 & 0
    \end{pmatrix}.
    \end{gathered}
\end{equation}
Applying Theorem~\ref{Theo::Better_Characterisation_of_Nash} yields the
equilibrium classes listed in Table~\ref{tab:three_action_classes}.

Now that we have characterised the equilibrium conditions and the corresponding
classes, a natural question arises. Since many equilibrium classes exist, which
ones are actually adopted in a social learning process?

To address this question, we consider a pairwise comparison process~\cite{traulsen:PRE:2006b}. 
This process considers a population of individuals with different strategies. 
Individuals receive payoffs by interacting in repeated games with randomly chosen opponents drawn from the population. 
Occasionally, they may revise their strategy by either adopting a random new strategy (akin to a mutation), or by imitating a strategy of another population member (akin to selection). 
A common assumption is that strategies with larger payoffs are more likely to be imitated~\cite{szabo:PRE:1998}, leading to the propagation of well-performing strategies within the population. 

To implement these dynamics, we use the process by Imhof and Nowak~\cite{imhof:PRSB:2010}.
This process assumes that mutations are rare, such that most of the time, everyone in the population adopts the same strategy~\cite{fudenberg:JET:2006,wu:JMB:2012,mcavoy:jet:2015}.  
Once a mutant strategy arises, it either fixes in the population or goes extinct before the next mutant arises. 
We consider this process in a population of $100$ individuals, initially all using the same resident strategy. 
At each of $10^8$ iteration steps, one individual then switches to a mutant strategy drawn uniformly at random
from the set of reactive strategies. 
Let $\pi_{M,j}$ denote the mean payoff of the mutant strategy in a population consisting of $j$ mutants and $100-j$
residents. Similarly, let $\pi_{R,j}$ denote the corresponding payoff of the resident
strategy. Then the fixation probability (the probability that the mutant
takes over the entire population) can be computed explicitly~\cite{nowak:Nature:2004}. 
It is defined by
\begin{align}
    \phi =
    \left(
    1+\sum_{i=1}^{100}
    \prod_{j=1}^{i}
    \exp\left(-\beta(\pi_{M,j}-\pi_{R,j})\right)
    \right)^{-1}.
\end{align}
Here, $\beta\!\ge\!0$ is the selection strength parameter. It determines how
strongly payoff differences affect the fixation probability. We record the
resident strategy throughout the evolutionary process and classify it according
to the equilibrium classes defined in
Theorem~\ref{Theo::Better_Characterisation_of_Nash}. Because each equilibrium
class is lower-dimensional than the full strategy space, the probability of
sampling a Nash equilibrium exactly is zero. To classify resident strategies
into one of the equilibrium classes, we therefore relax the equilibrium
conditions slightly.
First, we allow probabilities of up to $0.05$ on actions outside the support.
That is, we require
\begin{equation}
    p_b(a) < 0.05
\end{equation}
for all $a \in A\setminus S$ and $b \in S$.
Second, we allow small differences in payoffs against unconditional strategies.
Specifically, we require
\begin{equation}
    |\pi^i(\mathrm{All}{a_1},p)-\pi^i(\mathrm{All}{a_2},p)| < 0.2
\end{equation}
for all $a_1,a_2\in S$.

Figure~\ref{fig::evolutionary_sims}\textbf{(a)} shows the abundances of the
equilibrium classes as $b_1$ varies over the interval $[2.2,4.0]$, with
$b_2\!=\!2$, $c_1\!=\!1$, and $c_2\!=\!0.5$ fixed. We see that only a subset of the
equilibrium classes is observed frequently. For small values of $b_1$,
$M$-supporting equilibria are the most abundant. As $b_1$ increases, they are
replaced by $C$-supporting equilibria, which are more commonly known as partner
strategies~\citep{Hilbe:partners:2015}. The remaining equilibrium classes are
observed only rarely.

To understand these abundance patterns, we consider two possibilities. First, some
equilibrium classes may be sampled more frequently than others simply because
they occupy a larger volume of the strategy space. Second, equilibria may differ in
their robustness against mutant invasions. We therefore analyze both the dimensions of
the equilibrium classes and their invasibility in turn. 

We begin with the dimensions. Applying
Proposition~\ref{Prop::Degrees_of_freedom} to the three-action donation game
yields an upper bound of $4$ degrees of freedom for equilibrium classes with
$|S|\!=\!1$ and $|S|\!=\!3$, and an upper bound of $3$ degrees of freedom for classes
with $|S|\!=\!2$ (for comparison, the full reactive strategy space is
$6$-dimensional). Thus, equilibrium classes supported on two actions occupy
a smaller volume of the strategy space. 
The lower abundance of equilibrium classes with $|S|\!=\!2$ may already be
explained by this dimensional effect: Even if some of these equilibria are
robust against invasion, their smaller dimensionality makes them much less
likely to arise through mutation than equilibria belonging to higher-dimensional
classes.

To investigate the second mechanism, we analyze invasibility directly.
Figure~\ref{fig::evolutionary_sims}\textbf{(b)} shows the invasion probabilities of
randomly generated equilibria. For $b_1\!=\!3$ and $b_1\!=\!3.5$, we generate $100$
equilibria from each class satisfying the exact equilibrium conditions. Each
equilibrium is then challenged by $10^5$ random mutants drawn uniformly from the
full set of reactive strategies. We record how often an equilibrium from
each class is invaded.
For $b_1\!=\!3$, $M$-supporting equilibria are invaded least often. This is
consistent with their high abundance in the evolutionary simulations. For
$b_1\!=\!3.5$, the situation changes. $C$-supporting equilibria are invaded at a
similar rate to $M$-supporting equilibria, yet they are substantially more
abundant. This suggests that differences in invasibility alone cannot explain
the observed frequencies.

Figure~\ref{fig::evolutionary_sims}\textbf{(c)} provides additional insight. It shows the
probability that a randomly generated strategy satisfies the third condition of
Theorem~\ref{Theo::Better_Characterisation_of_Nash}, given that it satisfies the first two. For both values of $b_1$,
$C$-supporting equilibria are generated more easily than the competing classes.
Thus, their high abundance can be attributed both to their large dimension and
to the relatively weak restrictions imposed by the equilibrium conditions.

Finally, equalizers, corresponding to the class
$\{C,M,D\}$, have the same number of degrees of freedom as the
single-action equilibrium classes. Nevertheless, they are invaded more often
than any other class. Their high invasibility therefore appears to prevent them
from becoming common residents despite their comparatively large dimension.

Taken together, the simulations suggest that the evolutionary abundance of equilibrium
classes is determined by a combination of two effects: how easily equilibria
from a class can be generated through mutation and how robust they are against
invasion once they arise.

\section{Discussion}
In this paper, we characterised the set of reactive Nash equilibria in repeated
additive games with an arbitrary finite number of actions. We showed that the
equilibrium conditions can be expressed as a finite system of linear equalities
and inequalities in the strategy parameters. In addition, we derived a
representation of a reactive strategy's self-payoff that avoids explicit payoff
calculations and allows the equilibrium conditions to be obtained in a
tractable form.

Repeated games admit a rich set of equilibrium outcomes, as highlighted by the
Folk theorem and its variants~\citep{Fudenberg:Game:1991,mailath:book:2006}.
However, obtaining explicit descriptions of equilibrium behaviour remains
challenging. As a result, previous work has often focused on specific strategies
or on particular equilibrium
classes~\citep{Do:JTB:2017,Hilbe:PNAS:2017,Murase:SR:2020,Ueda:RSOS:2021,Ueda:ORF:2022,Li:NatCompSci:2022,Glynatsi:PNAS:2024}.
Moreover, much of this work has been restricted to games with only two actions.
Our results complement these approaches by providing an explicit
characterisation of the set of reactive Nash equilibria in repeated additive
games with an arbitrary finite number of actions.

A central finding is that the equilibrium set is naturally partitioned according
to the actions played at equilibrium. Specifically, there is a
one-to-one correspondence between equilibrium classes and the non-empty subsets
of the action set. An equilibrium belongs to the class associated with a subset
$S$ if, when facing itself, the strategy uses exactly the actions in $S$. We
call such an equilibrium $S$-supporting. Moreover, if an equilibrium is
$S$-supporting, then any opponent that uses only actions in $S$ against the
equilibrium strategy is also a best response, irrespective of how many previous
rounds the opponent remembers or how frequently the different actions in $S$ are
played. In the special case $S=A$, we recover the class of equalizer strategies,
for which every opponent obtains the same
payoff~\citep{boerlijst:AMM:1997,Press:PNAS:2012,Hilbe:equalizers:2013}. Thus,
equalizer strategies arise naturally as the equilibria with support on the whole
action set within our general characterisation.

The usefulness of our characterisation becomes apparent in the three-action
donation game. The equilibrium conditions remain simple and interpretable even
when additional actions are introduced (Table~\ref{tab:three_action_classes}).
Importantly, the derivation does not depend on the particular interpretation of
the third action. For example, if the benefit $b_2$ of the third action $M$ in
Table~\ref{tab:three_action_classes} is negative, the game can be interpreted as
a donation game with costly punishment. In this case, the conditions for
$C$-supporting equilibria become less restrictive, whereas $M$-supporting
equilibria cannot exist because condition (c) cannot be satisfied. Thus, our
conditions immediately imply that punishment-supporting equilibria cannot exist,
while cooperation-supporting equilibria become easier to satisfy. In this way,
our framework recovers previously studied settings as special
cases~\citep{Ohtsuki:Punishment:2009} while extending them to arbitrary additive
games with finitely many actions.

Once the equilibrium classes are characterised, a natural question is which of
them are likely to arise under evolutionary or learning dynamics. The number of
degrees of freedom of an equilibrium class provides one indication, since
classes with fewer degrees of freedom are generally less likely to be reached
by mutation or exploration. However, the prevalence of a class also depends on
its robustness to invasion. These two effects need not align. For example, in
our simulations of the three-action donation game, equilibria supported on the
complete action set (equalizers) appear to be more susceptible to invasion,
whereas equilibria supported on a single action are the most robust, despite
having the same number of degrees of freedom. Our results
(Figure~\ref{fig::evolutionary_sims}) suggest that differences in the abundance
of equilibrium classes are shaped by both how easily they are reached and how
robust they are to invasion.

On a technical level, an important contribution of this paper is the
representation of a reactive strategy's self-payoff in additive games. This
representation allows equilibrium conditions to be derived without explicitly
computing self-payoffs and is ultimately responsible for the tractability of
the problem. Beyond the present work, it would be interesting to understand
whether analogous payoff representations exist for richer strategy classes or
for non-additive games.

Although additive games do not encompass all repeated interactions, they include
many important models of social and biological behaviour, such as the donation
game~\citep{sigmund:book:2010} and games with
punishment~\citep{Ohtsuki:Punishment:2009} or reward~\citep{Fang:MPES:2019}.
Our results show that, within the class of additive games and reactive
strategies, the equilibrium set exhibits a surprisingly simple structure.

\section*{Data accessibility}
The source code used to reproduce the results of the evolutionary simulations
and generate the figures in this manuscript is available in the following
GitHub repository: \url{https://github.com/FranziLesi/reactive-nash-in-additive-games}.

\section*{Authors' contributions}
All authors designed research, performed research and wrote the paper.

\section*{Acknowledgements}
F.L. acknowledges support from the international internship program at RIKEN
R-CCS, where part of this study was conducted. C.H. acknowledge generous support
by the European Research Council Starting Grant 850529: E-DIRECT.

\section*{Artificial Intelligence (AI) declaration}
We used ChatGPT (OpenAI) to assist with proofreading and language editing of the
manuscript. All scientific content, analyses, interpretations, and conclusions
were generated and verified by the authors.

\section*{Appendix}\label{sec::Appendix}

We first establish an auxiliary result, which generalizes \emph{Akin's Lemma}~\cite{Akin:Ergodic:2016,Hilbe:partners:2015}.

\begin{Lemma}[Multi-action Akin]\label{lem::Akin}
    Let $p$ be a memory-1 strategy. Let the opponent be arbitrary. For $a,b \in A \times B$ let $\delta_{ab} : A \rightarrow \set{0,1}$ be defined by $\delta_{ab}(a_0) = 1 \Leftrightarrow a_0=a$. Then for all $a_0 \in A$
    \begin{equation}
        \sum_{(a,b) \in A\times B} \left(p_{ab}(a_0) - \delta_{ab}(a_0)\right) v(a,b) = 0
    \end{equation}
\end{Lemma}

\begin{proof}
    Let $p_t(a_0)$ be the probability that $p$ plays action $a_0$ in round $t$. Therefore 
    \begin{equation}
    p_t(a_0) = \sum_{(a,b) \in A\times B} \delta_{ab}(a_0) v_t(a,b)
    \end{equation} 
    and 
    \begin{equation}
        p_{t+1}(a_0)= \sum_{(a,b) \in A\times B} p_{ab}(a_0) v_t(a,b).
    \end{equation}
     We define $w(t) \coloneqq p_{t+1}(a_0) -p_{t}(a_0)$ and obtain that
    \begin{equation}
        w(t) = \sum_{(a,b) \in A\times B} \left(p_{ab}(a_0) - \delta_{ab}(a_0)\right) v_t(a,b).
    \end{equation}
    Further 
    \begin{align}
        \begin{aligned}
        \frac{1}{\tau +1} \sum_{t=0}^{\tau} w(t) &= \frac{1}{\tau +1}  \sum_{t=0}^{\tau} p_{t+1}(a_0) -p_{t}(a_0)\\
        &= \frac{1}{\tau +1} \left(\underbrace{p_{\tau+1}(a_0)}_{|.|\leq 1} -\underbrace{p_{0}(a_0)}_{|.|\leq 1}\right) \overset{\tau \rightarrow \infty}{\rightarrow} 0.
        \end{aligned}
    \end{align}
    On the other hand 
    \begin{align}
        \begin{aligned}
             \frac{1}{\tau +1} \sum_{t=0}^{\tau} w(t) &= \frac{1}{\tau +1} \sum_{t=0}^{\tau} \sum_{(a,b) \in A\times B} \left(p_{ab}(a_0) - \delta_{ab}(a_0)\right) v_t(a,b) \\
             &= \sum_{(a,b) \in A\times B} \left(p_{ab}(a_0) - \delta_{ab}(a_0)\right) \underbrace{\frac{1}{\tau +1} \sum_{t=0}^{\tau}  v_t(a,b)}_{\overset{\tau \rightarrow \infty}{\rightarrow} v(a,b)} \\
             &\overset{\tau \rightarrow \infty}{\rightarrow} \sum_{(a,b) \in A\times B} \left(p_{ab}(a_0) - \delta_{ab}(a_0)\right) v(a,b).
        \end{aligned}
    \end{align}
\end{proof}

This allows us to prove all lemmas from Section~\ref{sec:nash_equilibria}.

\begin{proof} [Proof of Lemma~\ref{lem::finding_unconditional}]
    Since both strategies are memory-$n$ strategies for some $n \in \N$ and we assume fixed initial actions, there exists a unique stationary distribution $w$ of the game. The states are the possible histories the strategies with the larger memory observes. We obtain the distribution of action frequencies $v$ then by summing together states that share the same last move. 
    
    \begin{align}
        v(a,b) = \sum_{(a_1b_1, \dots, a_{n-1}b_{n-1})\in (A\times B)^{n-1}} w((a_1b_1, \dots, a_{n-1}b_{n-1},ab))
    \end{align}
    That is, we group together all histories where the latest observed action pair is the same. The resulting distribution captions how frequently each action pair was played.

    Consider the action frequencies $v$ of the interaction of two arbitrary strategies $\sigma_1$ and $\sigma_2$. Their long-term payoffs in an additive game take the form 
    \begin{align}\label{eq::stationary_split_additive}
        \begin{aligned}
        (\pi^i(\sigma_1,\sigma_2), \pi^{-i}(\sigma_1, \sigma_2))^T &= \sum_{(a,b) \in A\times B} v(a,b) g(a,b)\\
         &= \sum_{(a,b) \in A \times B} v(a,b) \left(g_A(a) + g_B(b)\right) \\
        &= \sum_{a \in A} g_A(a) \left(\sum_{b \in B} v(a,b)\right) + \sum_{b \in B} g_B(b) \left(\sum_{a \in A} v(a,b)\right).
        \end{aligned}
    \end{align}
    It follows that given two different games with action frequencies $v_1$ and $v_2$, the long-term payoffs of both players are equal as long as for all $a \in A$ 
    \begin{equation}\label{eq::averaging_a}
        \sum_{b \in B} v_1(a,b) = \sum_{b \in B} v_2(a,b)
    \end{equation}
    and for all $b \in B$
    \begin{equation}\label{eq::averaging_b}
        \sum_{a \in A} v_1(a,b) = \sum_{a \in A} v_2(a,b).
    \end{equation}
    Denote by $\tilde{v}$ the action frequencies of $p$ playing against the unconditional strategy $q$. It remains to prove that Equations~\eqref{eq::averaging_a} and~\eqref{eq::averaging_b} hold for $v_1 = v$ and $v_2 = \tilde{v}$. 

    Consider $b_0 \in B$ arbitrary but fixed. We obtain from Lemma~\ref{lem::Akin}
    \begin{align}
        \begin{aligned}
        \sum_{(a,b) \in A\times B} q_{(a,b)}(b_0)\tilde{v}(a,b) &=  \sum_{(a,b) \in A\times B} \delta_{ab}(b_0) \tilde{v}(a,b)\\
        \sum_{(a,b) \in A\times B} q(b_0)\tilde{v}(a,b) &=  \sum_{a \in A} \tilde{v}(a,b_0)\\
        q(b_0) \underbrace{\sum_{(a,b) \in A\times B} \tilde{v}(a,b)}_{=1} &=  \sum_{a \in A} \tilde{v}(a,b_0)\\
        q(b_0) &=  \sum_{a \in A} \tilde{v}(a,b_0).
        \end{aligned}
    \end{align}
    Since $q(b_0) \coloneqq \sum_{a \in A} v(a,b_0)$ by definition, we obtain Equation~\eqref{eq::averaging_b}.
    
    Consider arbitrary but fixed $a_0 \in A$. Again, it follows from Lemma~\ref{lem::Akin} for $p$ playing against $q$ that
    \begin{align}
        \sum_{(a,b) \in A\times B} p_{ab}(a_0)\tilde{v}(a,b) =  \sum_{(a,b) \in A\times B} \delta_{ab}(a_0) \tilde{v}(a,b).
    \end{align}
    Since $p$ is reactive we obtain that
    \begin{equation}\label{eq::comparing_averaging_1}
            \sum_{b\in B} p_{b}(a_0) \sum_{a \in A} \tilde{v}(a,b) =  \sum_{b \in B} \tilde{v}(a_0,b).
    \end{equation}
    Similarly it follows for $p$ playing against $\sigma$ that
    \begin{equation}\label{eq::comparing_averaging_2}
            \sum_{b\in B} p_{b}(a_0) \sum_{a \in A} v(a,b) =  \sum_{b \in B} v(a_0,b)
    \end{equation}

    Due to Equation~\eqref{eq::averaging_b} the right-hand sides of Equations~\eqref{eq::comparing_averaging_1} and~\eqref{eq::comparing_averaging_2} must be equal. Hence, their left-hand sides must also coincide, which yields $\sum_{b \in B} \tilde{v}(a_0,b)=\sum_{b \in B} v(a_0,b)$. Since $a_0 \in A$ was arbitrary, Equation~\eqref{eq::averaging_a} follows.
\end{proof}

\begin{proof} [Proof of Lemma~\ref{lem::unconditional_payoff}]
    We first consider $\pi(p,q) = (\pi^i(p,q), \pi^{-i}(p,q))^T$. Since $q$ is unconditional it is also a self-reactive strategy. We thus obtain the following payoff representation, see~\cite{Lesigang:EconLett:2025},
    \begin{equation}
        \begin{aligned}
            \pi(p,q) = 
            \sum_{b \in B} q(b) \sum_{\substack{a \in A \\ \tilde{b} \in B}} g(a,\tilde{b}) p_b(a) q(\tilde{b}).
        \end{aligned}
    \end{equation}

    It follows from the game being additive that

    \begin{align}\label{eq:convex_comb_left}
        \begin{aligned}
            \pi(p,q) &= 
            \sum_{b \in B} q(b) \sum_{\substack{a \in A \\ \tilde{b} \in B}} \left(g_A(a) + g_B(\tilde{b})\right) p_b(a) q(\tilde{b}) \\
            &= \sum_{b \in B} q(b) \left(\sum_{a \in A }g_A(a) p_b(a) \underbrace{\sum_{\tilde{b} \in B}q(\tilde{b})}_{=1} + \sum_{\tilde{b} \in B} g_B(\tilde{b}) q(\tilde{b}) \underbrace{\sum_{a \in A}p_b(a)}_{=1}\right)\\
            &= \sum_{b \in B} q(b) \left(\sum_{a \in A }g_A(a) p_b(a) + \sum_{\tilde{b} \in B} g_B(\tilde{b}) q(\tilde{b}) \right)\\
            &= \left(\sum_{b \in B} q(b) \sum_{a \in A }g_A(a) p_b(a)\right) + \sum_{\tilde{b} \in B} g_B(\tilde{b}) q(\tilde{b})\underbrace{\sum_{b \in B} q(b)}_{=1} \\
            &= \left(\sum_{b \in B} q(b) \sum_{a \in A }g_A(a) p_b(a)\right) + \sum_{b \in B} g_B(b) q(b).
        \end{aligned}
    \end{align}
    For an arbitrary but fixed action $b_0$ it holds that
    \begin{equation}
    \pi(p,Allb_0)
         =
        \sum_{a \in A} p_{b_0} (a) \left(g_A(a) + g_B(b_0)\right).
    \end{equation}
    Thus the right hand side of~\eqref{eq::convex_comb} becomes
    \begin{align}\label{eq::convex_comb_right}
        \begin{aligned}
            \sum_{b\in B} q(b) \pi(p,Allb) &= 
            \sum_{b\in B} q(b) \sum_{a \in A} p_{b} (a) \left(g_A(a) + g_B(b)\right) \\
            &= \left(\sum_{b\in B} q(b) \sum_{a \in A} p_{b} (a) g_A(a)\right) + \sum_{b\in B} q(b) g_B(b) \underbrace{\sum_{a \in A} p_{b} (a)}_{=1}\\
            &= \left(\sum_{b \in B} q(b) \sum_{a \in A }g_A(a) p_b(a)\right) + \sum_{b \in B} g_B(b) q(b).
        \end{aligned}
    \end{align}
    This proves the theorem.
\end{proof}

\begin{proof}[Proof of Lemma~\ref{Lem::Characterisation_of_Nash}]
    $(\Leftarrow)$ Let $p$ be a reactive strategy that fulfills all four conditions. Consider $\pi^i(p,p)$ as obtained in Corollary~\ref{Kor::self-payoff}. Together with Condition~1 and Condition~2 it holds that
    \begin{align}
        \begin{aligned}
        \pi^i(p,p) &= \sum_{a \in A} \left(\sum_{\tilde{a} \in A} v(a, \tilde{a})\right) \pi^i(Alla,p) \\
        &= \sum_{s \in S} \left(\sum_{\tilde{s} \in S} v(s, \tilde{s})\right) \pi^i(Alls,p).
        \end{aligned}
    \end{align}
    Let $a_1$ be any arbitrary element of $S$. It follows from Condition~3 that
    \begin{align}
        \begin{aligned}
            \pi^i(p,p) &= \pi^i(Alla_1,p)\sum_{s \in S} \left(\sum_{\tilde{s} \in S} v(s, \tilde{s})\right) \\
            &= \pi^i(Alla_1,p).
        \end{aligned}
    \end{align}
    For any $a_2 \in A\setminus S$ Condition~4 implies that $\pi^i(p,p) \geq \pi^i(Alla_2,p)$. Since~\cite{Lesigang:EconLett:2025} shows that a best response exists among the pure unconditional strategies, we obtain that $p$ is a Nash Equilibrium.
    
    $(\Rightarrow)$ Define $S \coloneqq \argmax_{a \in A} \pi^i(Alla, p)$. Since $p$ is a Nash equilibrium it holds that 
    \begin{align}
        \pi^i(p,p) = \pi^i(Alla_1,p) \text{ for all } a_1 \in S
    \end{align}
    and by the definition of $S$ that
    \begin{align}
        \pi^i(p,p) > \pi^i(Alla_2,p) \text{ for all } a_2 \in A\setminus S
    \end{align}

    We consider again the self-payoff as established in Corollary~\ref{Kor::self-payoff}
    \begin{align}
        \pi^i(p,p) = \sum_{a \in A} \left(\sum_{\tilde{a} \in A} v(a, \tilde{a})\right) \pi^i(Alla,p).
    \end{align}
    Note that $\pi^i(p,p)$ is a convex combination over all constant strategies. Since $p$'s payoff matches the highest value among the constant strategies, all coefficients of constant strategies that do not achieve the highest value must be $0$, i.e.,
    \begin{align}
        \sum_{\tilde{a} \in A} v(a, \tilde{a}) = 0 \text{ for all } a \in A\setminus S.
    \end{align}
    Because $v$ is a distribution all components of $v$ are non negative. Thus $v(a,\tilde{a}) = 0$ for all $\tilde{a} \in A$. Due to symmetry of $v$ we obtain Condition~$1$. Condition~2 ensures that the stationary distribution is unique.
\end{proof}

\begin{proof} [Proof of Lemma~\ref{lem::Adopt_condition_3}]
    $(\Rightarrow)$ Let $a_1 \in S$ and $a_2 \in A\setminus S$ be arbitrary but fixed. From Lemma~\ref{lem::Akin} it follows that
    \begin{align}
        \begin{aligned}
            \sum_{\tilde{a} \in A} p_{\tilde{a}}(a_2)\sum_{a \in A}v(a,\tilde{a}) = \sum_{a \in A} v(a_2, \tilde{a})
        \end{aligned}
    \end{align}
    By assumption this reduces to
    \begin{align}
    \begin{aligned}
        \sum_{\tilde{a} \in S} p_{\tilde{a}}(a_2)\sum_{a \in S}v(a,\tilde{a}) = 0
    \end{aligned}
    \end{align}
    Since all terms in the left hand side are non negative we obtain for arbitrary $a_1 \in S$
    \begin{align}
        p_{a_1}(a_2) \sum_{a \in S}v(a,a_1) = 0
    \end{align}
    We assume that $\sum_{a \in S}v(a,a_1)=0$. Then $S\setminus \set{a_1}$ would also fulfill all four conditions. This is a contradiction to $S$ being minimal. Thus $p_{a_1}(a_2) = 0$.

    $(\Leftarrow)$ For the other direction we consider the Markov chain graph.
    We group the states by $S\times S$ as well as $(A\times A)\setminus (S\times
    S)$. Consider arbitrary states $(a_1,a_2) \in S\times S$ and $(b_1,b_2) \in
    (A\times A)\setminus (S\times S)$. The transition probability between the
    two states is by definition $p_{a_1}(b_2)p_{a_2}(b_1)$. Because $(b_1,b_2)
    \in (A\times A)\setminus (S\times S)$ either $b_1 \notin S$ or $b_2 \notin
    S$ or both. Thus $p_{a_1}(b_2)p_{a_2}(b_1)=0$. Since we either assume a
    unique stationary distribution or that $p$'s initial move ensures that $p$
    eventually plays actions in $S$ we obtain that no state $(b_1,b_2)\in
    (A\times A)\setminus (S\times S)$ can be reached in the long run. Thus the
    closed class of the Markov chain is a subset of $S\times S$ which is
    equivalent to $v(b_1,b_2) = 0$. This proves the Lemma. 
\end{proof}

\bibliographystyle{naturemag}

\bibliography{bibliography}

\newpage

\subsection*{Figures and tables}

~\\[1cm]
\noindent
\begin{table}[htbp!]
\centering
\renewcommand{\arraystretch}{1.4}
\begin{tabular}{c l}
\hline
$S$ & Conditions\\\hline
$\{C\}$ &
$\begin{array}[c]{l}
\text{(a)} \; p_C=(1,0,0)\\
\text{(b)} \; b_1p_M(C)+b_2p_M(M)-c_2 \le b_1-c_1\\
\text{(c)} \; b_1p_D(C)+b_2p_D(M)\le b_1-c_1
\end{array}$
\\
\hline

$\{M\}$ &
$\begin{array}[c]{l}
\text{(a)} \; p_M=(0,1,0)\\
\text{(b)} \; b_1p_C(C)+b_2p_C(M)-c_1 \le b_2-c_2\\
\text{(c)} \; b_1p_D(C)+b_2p_D(M)\le b_2-c_2
\end{array}$
\\
\hline

$\{D\}$ &
$\begin{array}[c]{l}
\text{(a)} \; p_D=(0,0,1)\\
\text{(b)} \; b_1p_C(C)+b_2p_C(M)-c_1 \le 0\\
\text{(c)} \; b_1p_M(C)+b_2p_M(M)-c_2 \le 0
\end{array}$
\\\hline
$\{C,M\}$ &
$\begin{array}[c]{l}
\text{(a)} \; p_C=(p_C(C),1-p_C(C),0)\\
\text{(b)} \; p_M=\left(p_C(C)-\frac{c_1-c_2}{b_1-b_2},
1-p_C(C)+\frac{c_1-c_2}{b_1-b_2},0\right)\\
\text{(c)} \; b_1p_D(C)+b_2p_D(M)+p_C(C)(b_2-b_1)\le b_2-c_1
\end{array}$
\\\hline
$\{C,D\}$ &
$\begin{array}[c]{l}
\text{(a)} \; p_C=(p_C(C),0,1-p_C(C))\\
\text{(b)} \; p_D=\left(p_C(C)-\frac{c_1}{b_1},
0,1-p_C(C)+\frac{c_1}{b_1}\right)\\
\text{(c)} \; b_1p_M(C)+b_2p_M(M)-p_C(C)b_1\le -c_1+c_2
\end{array}$
\\\hline
$\{M,D\}$ &
$\begin{array}[c]{l}
\text{(a)} \; p_M=(0,p_M(M),1-p_M(M))\\
\text{(b)} \; p_D=\left(0,p_M(M)-\frac{c_2}{b_2},
1-p_M(M)+\frac{c_2}{b_2}\right)\\
\text{(c)} \; b_1p_C(C)+b_2p_C(M)-p_M(M)b_2\le -c_2+c_1
\end{array}$
\\\hline
$\{C,M,D\}$ &
$\begin{array}[c]{l}
\text{(a)} \; p_C=(p_C(C),p_C(M),1-p_C(C)-p_C(M))\\
\text{(b)} \; b_1p_C(C)+b_2p_C(M)-c_1
=
b_1p_M(C)+b_2p_M(M)-c_2\\
\text{(c)} \; b_1p_C(C)+b_2p_C(M)-c_1
=
b_1p_D(C)+b_2p_D(M)
\end{array}$\\\hline
\end{tabular}
\caption{\textbf{An overview of the $S$-supporting equilibria in the
3-action donation game.}
For each subset $S$ of actions, we depict the conditions required for a strategy of the respective strategy class to be an equilibrium, as given by Theorem~\ref{Theo::Better_Characterisation_of_Nash}.
}
\label{tab:three_action_classes}
\end{table}

%%%%%%% Figures %%%%%%%

\begin{figure}[htbp!]
    \centering
    \includegraphics[width=0.75\textwidth]{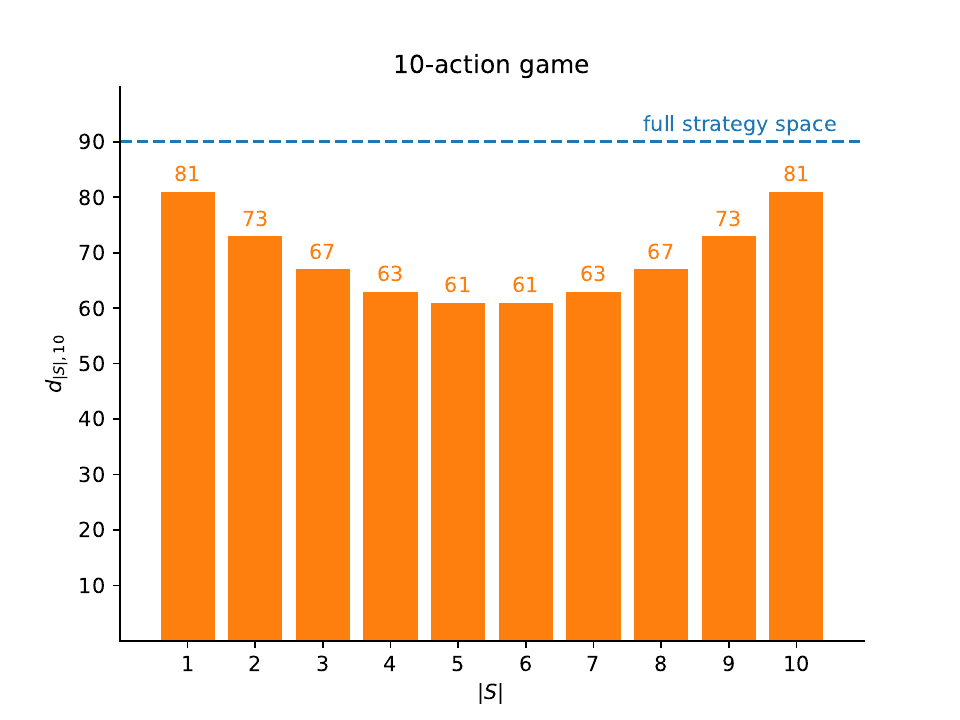}
    \caption{\textbf{Upper bound on the degrees of freedom of $S$-supporting equilibrium
    classes in a $10$-action game}. For each subset size $|S|$, we plot the upper
    bound from Proposition~\ref{Prop::Degrees_of_freedom}. The dimensions are
    symmetric around $|S|=(|A|+1)/2$, with the largest values attained at $|S|=1$
    and $|S|=|A|$. The dashed (blue) line indicates the dimension of the full reactive
    strategy space.}
    \label{fig::degrees_of_freedom_ten_actions}
\end{figure}

\begin{figure}
    \centering
    \includegraphics[width=\textwidth]{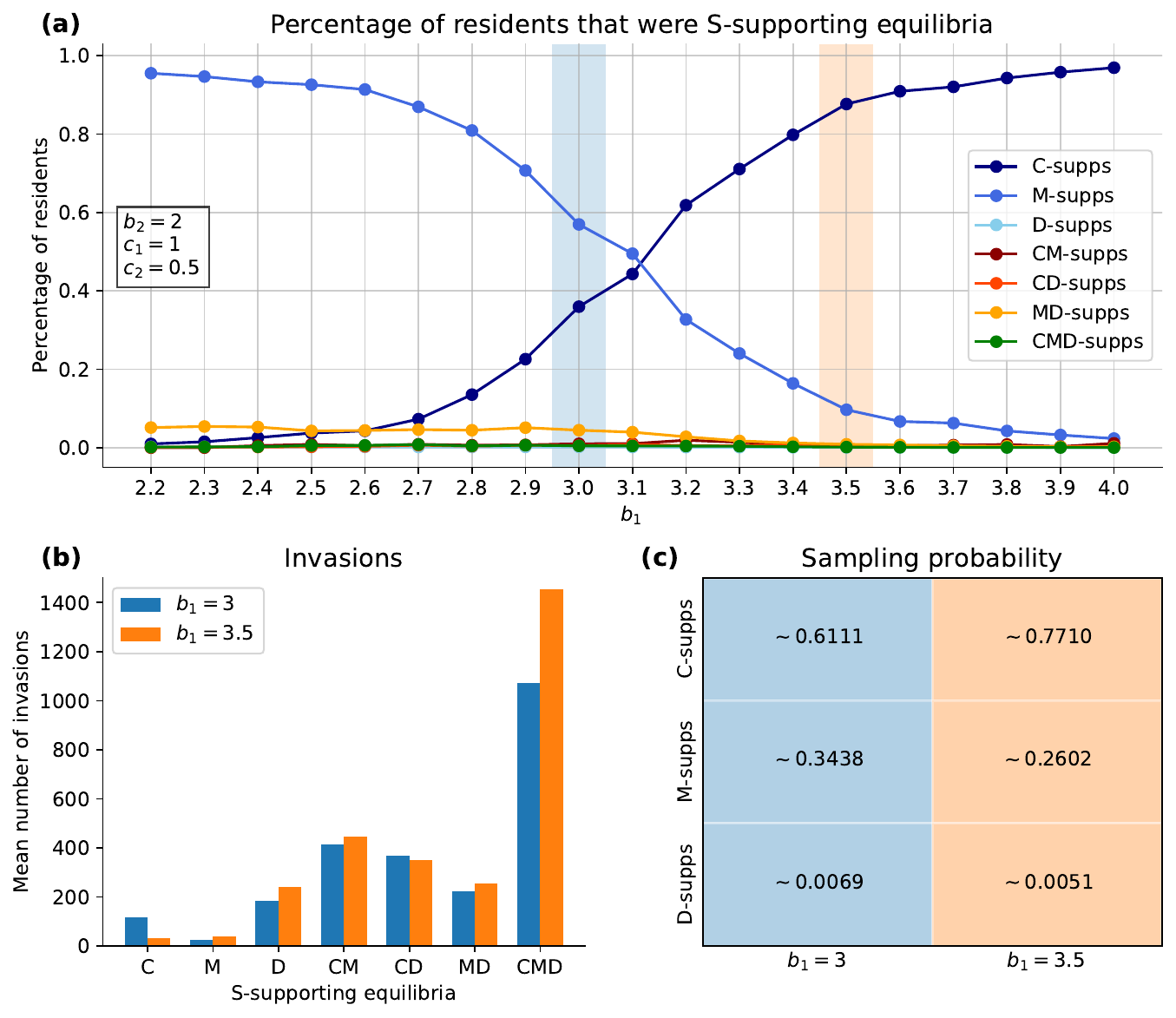}
    \caption{\textbf{Evolutionary dynamics of reactive strategies.}
    \textbf{(a),} We run twenty
    independent simulations and record the percentage of residents being
    $S$-supporting equilibria. For a strategy to be an $S$-supporting
    equilibrium we relax Condition~$1$ of
    Theorem~\ref{Theo::Better_Characterisation_of_Nash} by $0.05$ and
    Condition~$3$ by $0.2$. We observe that $M$-supporting equilibria are most
    abundant but $C$-supporting equilibria take over for increasing $b_1$.
    \textbf{(b),}
    Invasion stability for $S$-supporting equilibria. We generate $100$ random
    equilibria from each class and $10^5$ random reactive mutants for each
    equilibrium strategy. We track how often mutants invade the equilibrium
    strategy. We plot the average for each class. The high resistance of
    $M$-supporting equilibria and $C$-supporting equilibria could explain their
    abundance in the evolutionary simulations.
    \textbf{(c),} Abundance of $S$-supporting
    equilibria for $|S|=1$. We calculate the probability of a strategy being a
    $C$-supporting equilibrium given that $p_C(C) = 1$ and similarly for
    $M$-supporting and $D$-supporting equilibria. The fact that $D$-supporting
    equilibria are not observed in the evolutionary simulations could be
    explained by their low sampling probability. The simulations are based on
    the three-action donation game with $b_2 = 2$, $c_1 = 1$, and $c_2 = 0.5$ as
    well as $b_1 \in [2.2,4.0]$. The population size is $N =100$ and the
    selection strength is $\beta = 1$. We run our simulations for $10^8$ time
    steps.}
    \label{fig::evolutionary_sims}
\end{figure}

\end{document}